\documentclass{article}

\usepackage[numbers]{natbib}

\usepackage[ruled]{algorithm2e}

\SetAlFnt{\small}
\SetAlCapFnt{\small}
\SetAlCapNameFnt{\small}
\SetAlCapHSkip{0pt}
\IncMargin{-\parindent}

\usepackage{fullpage}
\usepackage{amssymb}
\usepackage{amsmath}
\usepackage{amsthm}
\usepackage{amsfonts}
\usepackage{color}
\usepackage{url}

\newcommand{\remove}[1]{#1}
\newcommand{\citeN}[1]{\cite{#1}}

\newtheorem{theorem}{Theorem}
\newtheorem{example}{Example}
\newtheorem{definition}{Definition}

\newtheorem{claim}{Claim}

\newtheorem{lemma}{Lemma}

\newtheorem{obs}{Observation}
\newtheorem{note}{Note}

\renewcommand{\epsilon}{\varepsilon}

\def\calD{{\cal D}}

\def\opt{\mbox{\rm opt}}

\def\re{\mathbb{R}}

\def\expectation{\mbox{\rm \bf E}}
\def\variance{\mbox{\rm \bf Var}}

\def\argmax{{\mbox{\rm argmax}}}

\def\out{{\mbox{\rm out}}}
\def\inf{{\mbox{\rm inf}}}

\begin{document}


\title{Privacy-Aware Mechanism Design}

\author{Kobbi Nissim\thanks{Dept.\ of Computer Science, Ben-Gurion University of the Negev, Be'er Sheva, Israel. {\tt kobbi@cs.bgu.ac.il}. Work partly done while the author was visiting the cryptography group at Bar-Ilan University.}  \and  Claudio Orlandi\thanks{Dept.\ of Computer Science, Bar-Ilan University, Ramat Gan, Israel. {\tt claudio.orlandi@cs.biu.ac.il}} \and 
Rann Smorodinsky\thanks{Faculty of Industrial Engineering and Management, Technion -- Israel Institute of Technology, Haifa 32000, Israel. {\tt rann@ie.technion.ac.il}.}} 

\maketitle

\begin{abstract}

Mechanism design deals with distributed algorithms that are executed with self-interested agents. The designer's, whose objective is to optimize some function of the agents private {\em types}, needs to construct a computation that takes into account agent incentives which are not necessarily in alignment with the objective of the mechanism. 
Traditionally, mechanisms are designed for agents who only care about the utility they derive from the mechanism outcome. This outcome often fully or partially discloses agent declare types. Such mechanisms may become inadequate when agents are privacy-aware, i.e., when their loss of privacy adversely affects their utility. In such cases ignoring privacy-awareness in the design of a mechanism may render it not incentive compatible, and hence inefficient. Interestingly, and somewhat counter-intuitively, Xiao [eprint 2011] has recently showed that this can happen even when the mechanism preserves a strong notion of privacy.
Towards constructing mechanisms for privacy-aware agents, we put forward and justify a model of privacy-aware mechanism design. We then show that privacy-aware mechanisms are feasible. The following is a summary of our contributions:
\begin{itemize}
\item {\bf Modeling privacy-aware agents:} We propose a new model of privacy-aware agents where agents need only have a conservative upper bound on how loss of privacy adversely affects their utility. This is in deviation from prior modeling which required full characterization.
\item {\bf Privacy of the privacy loss valuations:} Privacy valuations are often sensitive on their own. Our model of privacy-aware mechanisms takes into account the loss of utility due to information leaked about these valuations.
\item {\bf Guarantees for agents with high privacy valuations:} As it is impossible to guarantee incentive compatibility for agents that have arbitrarily high privacy valuations, we require a privacy-aware mechanism to set a threshold such that the mechanism is incentive compatible w.r.t.\ agents whose privacy valuations are below the threshold, and differential privacy is guaranteed for all other agents.
\item {\bf Constructing privacy-aware mechanisms:} We first construct a privacy-aware mechanism for a simple polling problem, and then give a more general result, based on recent generic construction of approximately additive mechanisms by Nissim, Smorodinsky, and Tennenholtz [ITCS 2012]. We show that under a mild assumption on the distribution of privacy valuations (namely, that valuations are bounded for all but a diminishing fraction of the population) these constructions are incentive compatible w.r.t.\ almost all agents, and hence give an approximation of the optimum. Finally, we show how to apply our generic construction to get a mechanism for privacy-aware selling of digital goods.
\end{itemize}

\end{abstract}







\medskip

\renewcommand{\thefootnote}{\arabic{footnote}}
\setcounter{footnote}{0}

\section{Introduction}

Mechanism design deals with distributed algorithms that are executed with self-motivated agents who optimize their own objective functions. The mechanism designer, interested in computing some function of the agents' private inputs (henceforth {\rm types}), needs hence to construct a computation that takes into account the agents' incentives, which are not necessarily in alignment with the goals of the designer. Settings where mechanism design is instrumental include centralized allocation of resources, pricing, the level of provision of a public good, etc. Traditionally, agents are modeled to care about the utility they derive from the outcome of the mechanism, but not about their privacy. Consequently, in many cases the outcome of the mechanism fully discloses the types declared by (some or all) agents.

We look at a model where agents also assign non-positive utility to the leakage of information about their private types through the public outcome of the mechanism.
This modeling is relevant, e.g., when private information is aggregated via markets which provide superior prediction power (e.g.,~\cite{WZ04}), kidney exchange markets where information aggregation and sharing lead to huge health-care benefits (e.g.,~\cite{ashlagi}), or recommendation engines which assist individuals in locating optimal products. Such markets may not be incentive compatible and consequently can fail if agents' privacy is not accounted for.

Our work is on the interface of the research in Algorithmic Game Theory and the recent theoretical research of privacy. Earlier scholarly work by \remove{McSherry and Talwar}~\citeN{MT07} has forged a link between the notion of {\em differential privacy}~\cite{DMNS06} and mechanism design. They observed that differential privacy can serve as a tool for constructing mechanisms where truthfulness is $\epsilon$-dominant. A recent work~\cite{NST12} has observed a few weaknesses in constructions resulting from~\cite{MT07} and resolved them by putting forward a general framework for constructing approximately-optimal mechanisms where truthfulness is a dominant strategy or an ex-post Nash equilibrium. This line of work demonstrates that differential privacy can serve as a powerful {\em tool} for the construction of efficient mechanisms.

The mechanisms presented in~\cite{MT07,NST12} were not analyzed with respect to agents who take into account their dis-utility due to the information leaked about their types. We call this dis-utility \emph{information utility} and we call {\em privacy-aware agents} those agents that take the information utility into account. It might be tempting to think that the combination of truthfulness and differential privacy is always sufficient for making privacy-aware agents truthful -- mechanisms that are truthful and preserve differential privacy should remain truthful also with respect to agents that take information utility into account. A work of \remove{Xiao}~\citeN{Xiao11} dispels this intuition by showing a mechanism that preserves differential privacy and is truthful with respect to agents that are not privacy aware, yet, under what seems to be a reasonable definition of information utility, truthfulness is not dominant with respect to privacy-aware agents. 

A recent work of \remove{Ghosh and Roth}~\citeN{GR11} constructs mechanisms that compensate agents for their loss in privacy. Ghosh and Roth consider a setting where a data analyst wishing to perform a differentially private computation of a statistic pays the participating agents for using their data. They construct mechanisms where agents declare how their loss of utility depends on the privacy parameter, and the mechanism decides upon which agents' information will be used in the computation and how much they will be paid. Interestingly, the mechanisms presented in~\cite{GR11} do not preserve the privacy of the loss valuations. However an agent value for privacy can reveal information about the agents' private data: it is not unreasonable to assume that there is some correlation between the price and agents sets on her privacy and the unlikelihood of her private data or, in other words, to assume that people value their privacy more if they have something to hide. 

In light of these issues, our goal is to construct mechanisms for privacy-aware agents, where privacy is accounted for the `traditional' inputs to the mechanism (such as valuations, locations, etc.) but also, and for the first time to the best of our knowledge, \emph{with respect to the privacy valuation itself}.

The results of~\cite{GR11} show, however, that this goal is too ambitious -- no individually rational mechanism can compensate individuals for the information (dis)utility incurred due to information leaked about the privacy valuation from the public output unless the privacy valuations are bounded. To overcome this obstacle we focus on mechanisms for large populations of agents: We propose a relaxation where loss in privacy is accounted for all agents whose valuations are bounded, where the bound increases as the agent population grows. Hence, in large enough populations truthfulness is provided for all (or most of) the agents. For the small fractions of agents who value their privacy too much for the mechanism to compensate,  we provide $\epsilon$-differential privacy with respect to whether their privacy valuations exceed the bound. The value of $\epsilon$ improves (i.e., reduces) with the population size.

\subsection{Our Contributions} 

\paragraph{Modeling} 

The main contribution of this work is a new notion of privacy-aware mechanism design where we examine critically previous modelings propose and justify a new model for privacy aware agents. 

We model privacy-aware agents to hold a `traditional' {\em game type} and a {\em privacy type}, where for the latter agents need only have a conservative upper bound on how loss of privacy adversely affects their utility. Agents care about leakage of information on both their game and privacy types. These features are in an important difference with respect to previous work (e.g.,~\cite{GR11}) where a full characterization of the information utility was required to achieve truthfulness, and furthermore, mechanisms did not take into account the information cost of the privacy type.

Note that if agents can have arbitrarily high privacy valuations, then it is impossible to a priori bound the information of a computation whose outcome depends on agents' private inputs, or, alternatively, on their  choice whether to participate or not (see also a more elaborate argument in~\cite{GR11} in the specific context of mechanisms for selling private information for statistical computations). To sidestep this inherent difficulty, we opt for a lesser requirement from a privacy-aware mechanism: the mechanism should set a threshold on the privacy valuation $v_{max}$ and a privacy parameter $\epsilon$ such that the mechanism is incentive compatible w.r.t.\ agents whose privacy valuations are below $v_{max}$ and $\epsilon$-differential privacy is guaranteed for all agents.

\paragraph{Construction of Privacy-Aware Mechanisms} 

We next demonstrate that privacy-aware mechanisms are feasible. Our first result illustrates some of our techniques: in Section~\ref{warmup}, we provide a simple privacy-aware poll between two or more alternatives. The main idea is to make (traditional) dis-utility due to mis-reporting dominate the information utility, and hence preserve truthfulness. We set a bound $v_{max}$ on the privacy valuations, and treat agents differently according to whether their valuations are above $v_{max}$ or not: for agents whose privacy valuations is below the bound, the mechanism ensures that the agents are provided with a fair reimbursement for their privacy loss. For agents whose privacy valuations are too high for the mechanism to compensate, we provide that their privacy valuations are protected in a $\epsilon$-differentially private way. As discussed above, this is in a sense the best we can hope to achieve. We then move our attention to large populations and we introduce the notion of \emph{admissible populations} by making a somewhat mild assumption on the distribution of the valuations (i.e., finiteness of its moments). 

In Section~\ref{sec:generic} we present a generic construction of privacy-aware mechanism. Our construction is based on the recent construction of~\cite{NST12}, where we modify the mechanism and its analysis to accommodate privacy-agents sampled from an admissible population. We show that the mechanism achieves truthfulness for most agents and non-trivial accuracy. Finally in Section~\ref{sec:mechanism digital goods} we present a natural example of a privacy-aware mechanism that falls in our framework i.e., privacy-aware selling of digital goods.

In a sense, our results show that when the outcome of a truthful (not necessarily privacy-aware) mechanism is insensitive to each of its individual inputs (as is often the case when the underlying population is large), it is rational for most privacy-aware agents to report truthfully. This is because the information leaked about their private types is small, and hence bounded away from the decrease in utility that can be caused by misreporting their type.

\subsection{Other Related Work}

The cryptographic literature also includes references to ``privacy preserving mechanism design'' (an example is \remove{Naor, Pinkas and Sumner}~\citeN{NPS99}). We stress that our goals are different from these cryptographic realizations of mechanisms as in our setting the agents are worried about what the public outcome of a mechanism may leak about their types and privacy valuations, whereas the goal of cryptographic realizations of mechanisms is to hide all information except for the outcome of the mechanism. As showed in~\cite{MNT09}, using cryptography to implement mechanism designs over an internet-like network is a non-trivial task, and one needs to make sure that the properties of the mechanism (e.g., truthfulness) are preserved also by the cryptographic implementation of the mechanism.

Independently from our work, \remove{Chen, Chong, Kash, Moran, and Vadhan}~\citeN{CCKMV} also studied the problem of truthful mechanisms in the presence of agents that value privacy. The motivation for both their work and ours is similar, and in both the quantification of privacy loss corresponds to the effect an agent's input has on the outcome of a mechanism. 
The model \cite{CCKMV} present for privacy-aware agents (and hence privacy-aware mechanisms) is different from ours in that in~\cite{CCKMV} agents are assumed to value privacy on a {\em per-outcome} basis, whereas our modeling utilizes a weaker assumption about the agents, i.e., that their privacy valuations depend on the {\em overall} (i.e., worst) outcome of the mechanism. Both modelings are well motivated, our reliance on a weaker assumption may lead to more robust mechanisms, where the per-outcome approach may lead to a richer set of privacy-aware mechanisms. 

\section{Preliminaries}

We refer to discrete sets $T$ and $S$ as the {\em type} set, and the  set of {\em social alternatives} respectively.
For two vectors $t,t'\in T^n$ we define the {\em Hamming distance} between $t$ and $t'$ as the number of entries on which $t,t'$ differ, i.e., $|\{i:t_i\not=t'_i\}|$. Vectors that are within Hamming distance one are called {\em neighboring}. 
A mechanism $M:T^n \rightarrow \Delta(S)$ is a function that assigns for any vector of inputs $t\in T^n$ a distribution over $S$ (the notation $\Delta(S)$ denotes the set of probability distributions over the set $S$). The outcome of an execution of $M$ on input $t\in T^n$ is an element $s\in S$ chosen according to the distribution $M(T)$.

\begin{definition}[Differential Privacy~\cite{DMNS06}]\label{def:differential privacy}
A mechanism $M: T^n \rightarrow \Delta(S)$ preserves $\epsilon$-differential privacy if for all neighboring $t,t'\in T^n$  and for all (measurable) subsets $S'$ of $S$ it holds that 
$$
M(t)(S') \leq e^\epsilon\cdot M(t')(S').
$$
\end{definition}

The following simple lemma follows directly from the above definition (the proofs for Lemma~\ref{lem:indifference} and Theorem~\ref{thm:expdp} below are not new and are included for completeness in Appendix~\ref{Missing Proofs}):

\begin{lemma}\label{lem:indifference}
Let $M: T^n \rightarrow \Delta(S)$ be a mechanism that preserves $\epsilon$-differential privacy and let $g: S\rightarrow \re^{\geq 0}$. Then, for all neighboring $t,t'\in T^n$
$$
\expectation_{s\sim M(t)} [g(s)] \leq e^\epsilon \expectation_{s\sim M(t')} [g(s)].
$$
In particular, if $\epsilon \leq 1$ and $g: S\rightarrow [0,1]$,
$$
\left|\expectation_{s\sim M(t)} [g(s)] - \expectation_{s\sim M(t')} [g(s)]\right| <  2\epsilon.
$$
\end{lemma}

A simple corollary of Lemma~\ref{lem:indifference} is that 
$\left|\expectation_{s\sim M(t)} [\hat g(s)] - \expectation_{s\sim M(t')} [\hat g(s)]\right| <  4\epsilon$ for neighboring $t,t'$ and $\hat g: S\rightarrow [-1,1]$.

\begin{definition}[\cite{MT07}]\label{def:expmech}
Let $f:S\times T^n \rightarrow \re^{\geq 0}$ and let $\epsilon >0$. The exponential mechanism for $f$ with parameter $\epsilon$ is 
$$M^\epsilon_f(t)(s) = \frac{\exp(\epsilon f(s,t))}{\sum_{s'\in S} \exp(\epsilon f(s',t))} \quad\mbox{for all}~s\in S.$$
\end{definition}

\begin{theorem}[\cite{MT07}]\label{thm:expdp}
Let $\Delta f$ be the maximum over all $s\in S$ and neighboring $t,t'\in T^n$ of $f(s,t)-f(s,t')$.
$M^{\frac{\epsilon}{2\Delta f}}_f$ preserves $\epsilon$-differential privacy.
\end{theorem}

\begin{definition}[Mutual Information]\label{def:mutual information}
Let $X,Y$ be two random variables. The mutual information between $X$ and $Y$ is defined as
$$
I(X;Y) = H(X) + H(Y) - H(X,Y),
$$
where $H(X) = - \sum_{x\in S} \Pr[X=x] \cdot \log \left(\Pr[X=x]\right)$ is the Shannon entropy of $X$.
\end{definition}
It is well known that $I(X;Y) = H(X) - H(X|Y)$, i.e. $I(X;Y)$ measures the reduction in entropy in $X$ caused by conditioning on $Y$ (and symmetrically, $I(X;Y) = I(Y;X) = H(Y)- H(Y|X)$). The following simple observation follows from the data processing inequality (see, e.g.,~\cite[pp.\ 32]{CT91}):
\begin{obs}\label{obs:data processing}
For all (randomized) functions $f$, 
$
I(f(X);Y) \leq I(X,Y).
$
\end{obs}

\begin{obs}\label{obs:bounds on I and H}
Let $M: T^n \rightarrow \Delta(S)$ be an $\epsilon$ differentially private mechanism, then for all random variables $X=(X_1,\ldots,X_n) \in \Delta(T^n)$ it holds that 
$
I(X_i; M(X), X_{-i}) \leq \epsilon. 
$ 
\end{obs}

\section{Quantifying Information Utility}\label{modelvi}

Our model is similar to the standard model of mechanism design, with the difference that agents participating in the execution of a mechanism care about their privacy. In the standard model, an agent's type $t_i$ expresses quantities such as a valuation of a good for sale, location, etc., the mechanism chooses an alternative $s$, and the agent's utility is a function of $t_i$ and $s$ (and sometimes, monetary transfers).

When considering privacy-aware agents, we need to introduce the information utility into their utility functions. A first issue that emerges is {\em how should this dis-utility be quantified?} Note that as different agents may value privacy differently, the quantification should be parametrized by agents' privacy preferences. We denote by $v_i$ the privacy preference of agent $i$. That is, an agent type is now composed of the `traditional' type $t_i$, and a privacy preference $v_i$. A second issue that now emerges is that the alternative chosen by the mechanism can leak information about both $t_i$ and $v_i$, and hence leakage about $v_i$ needs also be taken into account.

\paragraph{How is information utility quantified in prior work}

In an early work, \remove{McGrew, Porter, and Shoham} \citeN{MPS03} introduced privacy into agents' utility in the context of non-cooperative computing (NCC)~\cite{ST05}. In their model, agents only care about the case where other agents learn their private types with certainty. This means that privacy is either completely preserved or completely breached, and hence information utility is quantified to be either zero (no breach) or an agent dependent value $v_i>0$. As it is often the case that leaked information is partial or uncertain, we are interested in more refined measures that take partial exposure into account.

A recent work by \remove{Ghosh and Roth}~\citeN{GR11} considers a setting where a data analyst wishes to perform a computation that preserves $\epsilon$-differential privacy and compensates participating agents for their privacy loss. They assume a model where each agent's dis-utility is proportional to the privacy parameter $\epsilon$. I.e., the $i$th agent's dis-utility is 
$$
u^\inf_i = v_i\cdot\epsilon,
$$
where $v_i\geq 0$ is part of the agent's private type. A problem with this quantification is that while $\epsilon$ measures the worst effect the $\epsilon$-differentially private computation can have on privacy, the typical effect on agent $i$ can be significantly lower (see~\cite{DRV10}). Furthermore, 
it can depend on the other agents' inputs to the computation. Another problem, that will be further discussed later, is that this quantification does not consider the information utility due to leakage of information about $v_i$ itself.

The third example we are aware of is from another recent work, by \remove{Xiao}~\citeN{Xiao11}. Similarly to the present work, Xiao considers the setting of mechanism design with privacy-aware agents. The information utility is modeled to be 
$$
u^\inf_i = v_i \cdot I(t_i; M(t_{-i}, \sigma(t_i))),
$$
where $v_i\geq 0$ is the agent privacy valuation. Note that with this measure, the dis-utility of agent $i$ depends on the distribution of her and the other agents' types, and on her own strategy $\sigma$. The following example demonstrates that this dependency on $\sigma$ is problematic. 

Consider the single-agent mechanism below, where the agent's private type consists of a single bit:
\begin{example}[The ``Rye or Wholewheat'' game]\label{example:MutualInformationProblematic}
Alice is preparing a sandwich for Bob and inquires whether he prefers Rye (R) or Wholewheat (W). Bob wants to enjoy his favorite sandwich, but does not want Alice to learn his preference. Assume that Bob's type is uniformly chosen in $\{\mbox{\rm R},\mbox{\rm W}\}$ and consider these two possibilities for Bob's strategy:
\begin{enumerate}
\item If Bob provides his true preference he will enjoy the sandwich. However, his information (dis)utility would be maximized as $I(t_{Bob};M(\sigma_{truthful}(t_{Bob})))=1$.
\item If Bob answers at random he will enjoy the sandwich with probability one-half.\footnote{This is equivalent to encrypting Bob's type using a {\em one time pad}.}  However, as his response does not depend on his preference no loss in privacy would be incurred, hence we get $$I(t_{Bob};M(\sigma_{random}(t_{Bob})))=0 \ .$$
\end{enumerate}
\end{example}

Note that since Bob's type is $\mbox{\rm R}$ or $\mbox{\rm W}$ with equal probability Alice's views of Bob's actions (and hence also the outcome of the mechanism) are distributed identically whether he uses $\sigma_{truthful}$ or $\sigma_{random}$. Hence, while mutual information differs dramatically between the two strategies -- suggesting that Bob is suffering a privacy loss due to Alice learning his type in one but not in the other -- it is impossible for Alice to distinguish between the two cases!

A few more words are in place regarding the source of this problem. First note that while the example demonstrates that $I(t_{Bob};M(\sigma(t_{Bob})))$ is a problematic as a measure of privacy it does not imply that $\sigma_{random}$ (nor the one time pad) is at fault (in fact, $\sigma_{random}$ provides Bob with perfect privacy even in a setting where Alice gets to know which strategy Bob uses, a guarantee $\sigma_{truthful}$ definitely does not provide).
What the example capitalizes on is the fact that the standard game-theoretic modeling does not rule out the possibility that Alice does not get to see what Bob's strategy is. In such situations, it can happen that the more robust $\sigma_{random}$ is an overkill, as it provides Bob with less utility. We hence argue that the notion of information cost should be free of making assumptions on Alice's knowledge of $\sigma$.

\subsection{Our Approach} 

We deviate from the works cited above as we do not present a new measure for information utility. We use a significantly weaker notion instead.
To motivate our approach, re-consider the measures discussed above. 

Looking first at the measure in~\cite{GR11}, i.e., $v_i\cdot\epsilon$, we note that while in $\epsilon$-differential mechanisms the ratio $\Pr[M(t)=s]/\Pr[M(t')=s]$ is bounded by $e^\epsilon$ for all neighboring $t,t'$ and $s$, it is plausible that the worst case behavior (i.e., outputting $s$ such that $\Pr[M(t)=s]/\Pr[M(t')=s] = e^\epsilon$) occurs with only a tiny probability. This suggest that while $v_i\cdot\epsilon$ may not be a good measure for information utility, it can serve as a good {\em upper bound} for this utility. Examining the measure in~\cite{Xiao11} and trying to avoid the problem demonstrated in Example~\ref{example:MutualInformationProblematic} above, we note that by Observation~\ref{obs:data processing} $I(t_i; M(t)) \geq I(t_i; M(t_{-i}, \sigma(t_i)))$ for all $\sigma$, hence, we get that $v_i \cdot I(t_i; M(t))$ is another plausible {\em upper bound} for information utility. 
Finally, taking into account Observation~\ref{obs:bounds on I and H} we get that $I(t_i; M(t)) \leq \epsilon$ and hence we choose to use $v_i\cdot \epsilon$ as it is the weaker of these bounds.

\begin{note}
We emphasize that although our usage of the term $v_i \cdot \epsilon$ is syntactically similar to that of~\cite{GR11}, our usage of this quantity is conceptually very different. In particular, while loss of privacy cannot be used in our constructions for deterring non-truthful agents, the constructions (and proofs) in~\cite{GR11} use the fact that the information utility is (at least) $v_i \cdot \epsilon$ for arguing truthfulness.
\end{note}

\begin{note}
Lemma~\ref{lem:indifference} supports using $v_i\cdot\epsilon$ as an upperbound for information utility in the follwing sense. An individual's concern about her privacy corresponds to a potential decrease in {\em future} utility due to information learned about her. an upper bound on information utility hence should correspond to this (potential) loss in future utility. By Lemma~\ref{lem:indifference}, the information contributed by individual $i$ affects the expectation of every non-negative (similarly, non-positive) function $g$ by {\em at most} a factor of $e^\epsilon$. Let $G_i: S\rightarrow \re$ describe how the future utility of individual $i$ depends on the outcome of $M$. By Lemma~\ref{lem:indifference}, the information utility of that individual is bounded by 
$$
\max_{t \in T^n}(e^\epsilon - 1)\cdot\expectation_{s\sim M(t)}\left|G_i(s)\right| \approx \epsilon\cdot \max_{t \in T^n}\expectation_{s\sim M(t)}\left|G_i(s)\right|,
$$
where the approximation holds for small $\epsilon$. See also a related discussion in~\cite{GR11}. 
\end{note}

\paragraph{Privacy of $v_i$} The mechanisms presented in~\cite{GR11} for selling private information do not protect the privacy of $v_i$ nor they account for the information (dis)utility generated by the leakage of $v_i$. It is further shown that with unbounded $v_i$s it is impossible to construct mechanisms that  compensate agents for their loss in privacy and achieve reasonable accuracy (in the sense that enough agents sell their information). 

Our mechanisms provide an intermediate solution. First, we provide $\epsilon$-differential privacy to {\em all} agents, where the guarantee is with respect to their combined type, i.e., $(t_i,v_i)$, and where $\epsilon$ decreases with the number of agents $n$. This means that privacy improves as $n$ grows.

Furthermore our constructions guarantee that truthfulness is dominant -- {\em taking information utility about the combined type $(t_i,v_i)$ into account} -- for all agents for which $v_i \leq v_{max}$, where under a very mild assumption on the distribution of $v_i$ the bound $v_{max}$ grows with $n$ and the fraction of agents for which $v_i > v_{max}$ decreases with $n$.

\section{The Model}

\paragraph{The Mechanism}
Let $S$ be a finite set of alternatives (a.k.a.\ social alternatives), let $T$ be a finite type set and consider a set of $n$ agents. We consider direct revelation mechanisms that given the declaration of agents about their types selects a social alternative $s\in S$ and makes $s$ public. To isolate loss of privacy due to publication of $s$ from other potential sources of leakage, we will assume that every other information (including, e.g., the agents' declared types and individual monetary transfers) is completely hidden using cryptographic or other techniques.

\paragraph{The Objective Function} The goal of the designer is defined via a real, non-negative objective function over the true types of the agents, $f(t,s)$ that needs to be optimized (by choosing $s$).
$$
f : T^n \times S \rightarrow [0,n\Delta f].
$$
Following~\cite{DMNS06,MT07} we define the sensitivity of $f$ to be 
$$
\Delta f = \max |f(\hat t,s) - f(\hat t', s)|
$$ 
where the maximum is taken over all neighboring $\hat t, \hat t' \in T^n$ and $s\in S$. We assume that for all $s$ the minimum value of $f(t,s)$ is $0$ and then, given that sensitivity is $\Delta f$ by a hybrid argument we get that $f\leq n \Delta f$.

\paragraph{Privacy-Aware Agents} We extend the traditional setting of selfish agents to include agents who care not only about their utility $u^\out_i$ from the outcome $s$ of the mechanism, but also about the (negative) information utility $u^\inf_i$ incurred from the leakage of information about their private type through the public output $s$. 

For simplicity, we consider a setting where the overall utility of an agent is the sum of the two:\footnote{Admittedly, this separation of the utility function is sometimes artificial. However, we find it conceptually helpful.} 
$$
u_i = u_i^\out - u_i^\inf.
$$

An agent's type $\tau_i$ is modeled by a pair $\tau_i=(t_i,v_i) \in T\times\re^{\geq 0}$, where $T$ is the ``traditional'' game type and $v_i$ is the privacy valuation of agent $i$. We emphasize that agents care about the privacy of the whole pair and the information utility  corresponds to the loss in privacy of both $t_i$ and $v_i$ (hence, one cannot simply publish $v_i$). The vectors $t=(t_1,\ldots,t_n)$ and $v=(v_1,\ldots,v_n)$ denote the types of all agents. Trying to maximize her utility, agent $i$ may hence act strategically and declare $\tau'_i=\sigma_i(\tau_i)=
(t'_i,v'_i)$ to $M$ instead of $\tau_i$.

The ``traditional'' game utility of agent $i$ is defined as $u_i^\out : T \times S \rightarrow [-1,1].$
Following our discussion above, we define $u_i^\inf : \re^{\geq 0} \rightarrow \re^{\geq 0}$, and the only assumption we make is that 
$$
u_i^\inf(v_i) \leq v_i\cdot \epsilon
$$ 
for $\epsilon$ being the parameter of the differentially private mechanism executed (i.e., 
$e^\epsilon = \max\left(M(t)(S)/M(t')(S)\right)$
where the maximum is taken over all neighboring $t,t'\in T^n$ and $S'\subseteq S$). Note that unlike $u_i^\out$ that only depends on the outcome of the mechanism, $u_i^\inf$ depends on the mechanism itself.

In our analysis we identify a subset of agents that we call {\em participating} for whom truthtelling is strictly dominant. 
A mechanism approximately implements $f$ if assuming that participating agents act truthfully (and other agents act arbitrarily) it outputs $s$ that approximately optimizes $f$.

\subsection{Warmup: A Privacy-Aware Poll}\label{warmup}

The following simple electronic poll will serve to illustrate some of our ideas:

\begin{example}[An Electronic Poll]\label{example:warmup}
An electronic publisher wishes to determine which of its $m\geq 2$ electronic magazines is more popular. Every agent is asked to specify her favorite magazine, i.e., $t_i \in [m]$, and will receive in exchange an electronic copy of it. For simplicity, we assume that agents' utility does not depend on the poll outcome.

Following our convention, we assume ideal cryptography here, that is, no information beyond the outcome of the poll is leaked. In particular, every agent receives the electronic magazine without anybody (including the publisher) knowing which magazine has been transferred. Agents, however, are privacy-aware, and hence take into account that the outcome of the poll itself reveals information about their preferences.
\end{example}

Denote by $t'$ the vector of agents' declarations. 
For $s\in [m]$ let $f(s,t')=|\{i | t'_i =s\}|$ and note that $\Delta f=1$. Consider the exponential mechanism $M=M^{\frac{\epsilon}{2}}_f$ as in Definition~\ref{def:expmech}. I.e., 
$$\Pr[M(t')=j] = \frac{e^{\epsilon n'_j/2}}{\sum_{\ell=1}^m e^{\epsilon n'_\ell/2}},$$ 
where $n'_j = f(j,t')$  is the number of agents who declared they rank magazine $j$ first. By Theorem~\ref{thm:expdp}, $M$ preserves $\epsilon$-differential privacy.

Note that if $n'_j \geq n'_\ell + k$ then 
$$\Pr[M(t')=\ell] \leq \frac{e^{\epsilon n'_\ell/2}}{e^{\epsilon n'_j/2}} \leq \frac{e^{\epsilon n'_\ell/2}}{e^{\epsilon (n'_\ell+k)/2}} = e^{-\epsilon k/2}.$$
Hence, 
$$\Pr[M(t')~\mbox{outputs}~\ell~\mbox{such that}~n'_\ell < \max_j n'_j - k] \leq (m-1)e^{-\epsilon k/2}.$$

The agent utilities are $u_i^\out-u_i^\inf$ where 
\begin{itemize}
\item $u_i^\out$ is the utility that the agent gains from receiving the magazine she specified she prefers. Note that this utility depends only on the declared type and it is maximized for $t_i$, the true type of the agent; we assume that $u^\out_i(t_i)-u^\out_i(t'_i) \geq g$ (Alternatively, the publisher does not care if agent $i$ reports $t'_i$ if $u^\out_i(t_i)-u^\out_i(t'_i) < g$).

\item $u_i^\inf \leq \epsilon \cdot v_i$ is the privacy loss from the mechanism. 
\end{itemize}
Note that $\epsilon < g/v_i$ suffices for making agent $i$ truthful: acting untruthfully agent $i$ will lose at least $g$ in $u^\out_{i}$ and gain no more than $\epsilon v_i$ in $u_i^\inf$. 
Denote by $n_j$ the number of agents who rank magazine $j$ first (note the difference from $n'_j$ that correspond to declared types). To demonstrate that the mechanism is efficient, we need to make some (hopefully reasonable) assumptions on the distribution of $v_i$. We explore three possibilities:

\paragraph{Bounded $v_i$} We begin with a simplified setting where we assume that there exists $v_{max} =O(1)$ such that $\forall i : v_i \leq v_{max}$. In this case it is enough to set $\epsilon < g/v_{max}=O(1)$ to make truthfulness dominant for {\em all} agents. Hence, assuming all agents are truthful, we get $n'_j=n_j$ for all $j\in[m]$ and hence the probability that $M(t')=M(t)$ outputs $\ell$ such that $n_\ell < \max_j n_j - k$ is bounded by $(m-1)e^{-\epsilon k/2}$.

Note that in this case the computation output leaks no information about the privacy valuations $v$.

\paragraph{Bounded $v_i$, Except for a Small Number of Agents} A more realistic setting allows for a small number of agents with $v_i > v_{max}$. We change the mechanism $M$ to also consider the reported $v'_i$ so that inputs from agents with $v'_i\geq v_{max}$ are ignored. Regardless of what agents with $v_i > v_{max}$ report, we call them {\em non-participating}.

As before, by setting $\epsilon< g/v_{max}$ we make truthfulness dominant for all agents with $v_i\leq v_{max}$. We can hence guarantee a non-trivial accuracy. Let $n_{np}$ be the number of non-participating agents. In the worst case, non-participating agents deflate the count of a popular magazine and inflate the count of an unpopular magazine, making it look more popular than it really is. Taking this into account, we get that 
$$\Pr[M(t')~\mbox{outputs}~\ell~\mbox{ such that}~n_\ell < \max_j n_j - k-2n_{np}]\leq (m-1)e^{-\epsilon k/2}.$$

Note that we lose truthfulness for non-participating agents. We do, however, guarantee $\epsilon$-differential privacy for these agents.

\paragraph{Large Populations} Assume we do not care if the mechanism does not output the most popular choice if it does not have significant advantage over the other, e.g., when $k+n_{np} = O(n^{\alpha})$ for some $0 < \alpha < 1$. This allows us to set $\epsilon(n) = n^{-\alpha}$ and hence truthfulness is dominant for agents with $v_i \leq g/\epsilon = v_{max}(n) \in  O(n^\alpha)$. Note that $v_{max}$ grows with $n$, hence we expect the fraction of non-participants $n_{np}/n$ to diminish with $n$. If $n$ is large enough so that the fraction of agents for which $v_i > v_{max}(n)$ is at most $1/n^{1-\alpha}$ then we get the desired accuracy.

As before, we lose truthfulness for non-participating agents, and only guarantee $\epsilon(n)$-differential privacy for the non-participating agents. Note, however, that the fraction of non-participating agents diminishes with $n$, and, furthermore, their privacy guarantee improves with $n$ (i.e., $\epsilon(n)$ decreases).

\subsection{Admissible Privacy Valuations}

In the rest of the paper we only focus on large populations (the analysis can be easily modified for the case where $v_i$ is bounded except for a small number of agents). We will design our mechanisms for ``nicely-behaving'' populations:

\begin{definition}[Admissible Valuations]\label{def:admissibleValuations}
A population of $n$ agents is said to have $(\alpha,\beta)$-{\em admissible valuations} 
if $$\frac{|\{i: v_i > n^\alpha\}|}{n} \leq n^{-\beta}.$$
\end{definition}

To partly justify our focus on admissible valuations, consider the case where $v_i$ are chosen, i.i.d., from some underlying distribution $\calD$ over $\re^{\geq 0}$.
\begin{definition}[Admissible Valuation Distribution]\label{def:admissibleDistributions}
A valuation distribution $\calD$ is called $(\alpha,\beta)$-{\em admissible} if $$\Pr_{v\sim \calD}[v > n^\alpha] =O(n^{-\beta}).$$
\end{definition}

Note that if $\calD$ has finite expectation, then (using Markov's inequality) $\Pr[v > n^\alpha] \leq \expectation[v]/n^\alpha =O(n^{-\alpha})$, and hence $\calD$ is $(\alpha,\beta)$-admissible for all $\beta \leq \alpha$. If $\calD$ has finite variance then (using Chebyshev's inequality) $\Pr[v > n^\alpha] \leq \variance[v]/(n^\alpha - \expectation[v])^2 = O(n^{-2\alpha})$, and hence $\calD$ is $(\alpha,\beta)$-admissible for all $\beta \leq 2\alpha$. 
More generally, consider the following simple generalization of Chebyshev's inequality to even $p$th moment:
$$
\Pr[|X -\expectation[X]| > t] = \Pr[(X -\expectation[X])^p > t^p] \leq \frac{\expectation \left[(X-\expectation[X])^p\right]}{t^p}.
$$
Using this inequality in the argument above we get that if $\calD$ has finite even $p$th moment then $\calD$ is $(\alpha,\beta)$-admissible for all $\beta \leq p\alpha$. 
We conclude that if $\calD$ has finite $p$th moment then $\calD$ is $(\alpha,1-\alpha)$-admissible for $\alpha \geq 1/(p+1)$.
In particular, if $\calD$ has finite moments of all orders then $\calD$ is $(\alpha,1-\alpha)$-admissible for all $\alpha \in(0,1)$. 

We can even consider a notion of \emph{strong admissibility}:
\begin{definition}[Strongly Admissible Valuation Distribution]
A valuation distribution $\calD$ is called $\alpha$-{\em strongly admissible} if $$\Pr_{v\sim \calD}[v > (\log n)^\alpha] = n^{-\omega(1)},$$ 
where $n^{-\omega(1)}$ denotes a function that is negligible in $n$. 
\end{definition}
For example, the Normal distribution is $\alpha$-strongly admissible. In our analysis, however, we only use the more conservative notion of admissibility as in definitions~\ref{def:admissibleValuations},~\ref{def:admissibleDistributions}.

These simple observations suggest that $(\alpha,1-\alpha)$-admissibility is a relatively mild assumption that would typically hold in large populations even for small values of $\alpha$. 

\subsection{The Privacy-Aware Poll with Admissible Valuations}

Returning to our example, let $\alpha$ be the smallest positive value such that the agent population can be assumed to be $(\alpha,1-\alpha)$-admissible. 
By setting $v_{max} = n^\alpha$ and $\epsilon = g/v_{max} = gn^{-\alpha}$ we get that $n_{np}\leq n \cdot n^{-(1-\alpha)} = n^\alpha$. 
Finally, setting $k = n^\alpha(\log n)^2\log m /g$ we get the following:

\begin{claim}\label{clm:admissiblepoll} 
 The probability that $M(t')$ outputs $\ell$ such that $n_\ell < \max_j n_j - 2k$ is negligible in $n$.
\end{claim}

\section{A Generic Construction of Privacy-Aware Mechanisms}\label{sec:generic} 

We now present a generic feasibility result for privacy-aware mechanisms. Our construction is based on the construction of~\cite{NST12}, where differential privacy is used as a tool for mechanism design. The hope is that existence of this generic construction, a relatively simple modification of~\cite{NST12}, is a signal that our model of privacy-aware mechanisms allows constructing mechanisms for many other tasks.

\paragraph{Reactions} We first change our model to incorporate the notion of {\em reactions} introduced in~\cite{NST12}.\footnote{While the standard game-theoretic modeling does not explicitly include reactions, in many settings their introduction is natural. We refer the reader to~\cite{NST12} for further discussion of this change in the standard model.} Traditionally, an agent's utility is a function of her private type and the social alternative, and the issue of how agents exploit the social choice is not treated explicitly. In~\cite{NST12} this choice was made explicit such that after a social choice is made agents need to take an action (denoted {\em reaction}) to exploit the social alternative and determine their utility. In~\cite{NST12} (and likewise in this work) allowing the mechanism to sometimes restrict the reactions of agents serves as a deterrent against non-truthful agents. 

Let $R$ be a finite set of reactions. We modify the definition of the utility from the outcome of the mechanism to 
$$
u_i^\out : T \times S \times R \rightarrow [-1,1].
$$
Given $t_i,s$ define $$r_i(t_i,s) = \argmax_{r\in R}(u_i^\out(t_i,s,r))$$ to be the optimal reaction for agent $i$ on outcome $s$. 

To illustrate the concept of reactions, consider a mechanism for setting a price for a unlimited supply good (such as in Example~\ref{example:pricing digital good} appearing below). Once the mechanism chooses a price $s$ the possible reactions  are {\rm \bf buy} (i.e., pay $s$ and get the good) and {\rm \bf not buy} (i.e., do not pay $s$ and do not get the good), and reactions are kept hidden by assuming payment and reception of the digital good using perfect cryptography. In this example agents reactions may be restricted to {\bf \bf buy} whenever they bid at least the selected price $s$, and {\rm \bf not buy} otherwise. 

\paragraph{Utility Gap} We assume the existence of a positive {\em gap} $g$ such that for all $t_i\not=t'_i$ there exist $s$ for which the optimal reactions are distinct, and, furthermore, $u_i^\out(t_i,s,r_i(t_i,s)) \geq u_i^\out(t_i,s,r_i(t'_i,s)) + g$. In many setting, a gap $g$ can be created by considering a discrete set of social choices. As in our polling example, an alternative interpretation of the gap $g$ may be that the mechanism designer not care if agent $i$ reports $t'_i$ if $u_i^\out(t_i,s,r_i(t_i,s)) < u_i^\out(t_i,s,r_i(t'_i,s)) + g$.

\subsection{The Construction}

Given a finite type set $T$, a finite set $S$ of alternatives and an objective function $f: T^n \times S \rightarrow \re$ with sensitivity $\Delta f$, we construct a mechanism for approximately implementing $f$.

Let $n$ be the number of agents, and let $\alpha$ be the smallest positive value such that the agent population can be assumed to be $(\alpha,1-\alpha)$-admissible. Let $v_{max}=n^\alpha$. The participating agents will be those with privacy valuations lower than $v_{max}$.
Choose $t_\bot \in T$ to be an arbitrary element of $T$. Non-participating agents will be asked to declare $t_\bot$.

Let $\delta \in [0,1], \epsilon>0$ be parameters to be set later. Agents are asked to declare $t_i$ if $v_i \leq v_{max}$ and $t_\bot$ otherwise. Let $t'_i$ be the declaration of agent $i$. On input $t'=t'_1,\ldots,t'_n$ the mechanism executes as follows:



\begin{algorithm}[h]
\SetAlgoNoLine
\KwIn{A vector of types $t' \in T^n$.}
\KwOut{A social choice  $s \in S$.}

$M$ executes $M_1$ with probability $1-\delta$ and $M_2$ otherwise, where $M_1,M_2$ are as follows:
\begin{description}
\item[Mechanism $M_1$] For all $s \in S$ and $t'\in T^n$, choose $s \in S$ according to the exponential mechanism $M^{\frac{\epsilon}{2 \Delta f}}_f(t')$. 
\item[Mechanism $M_2$] Choose $s \in S$ uniformly at random.
\end{description}

The mechanism $M$ also restrict all agents to their optimal reactions according to their declarations, i.e., $r_i(t'_i,s).$\footnotemark

\caption{The generic mechanism $M$.}
\label{alg:one}
\end{algorithm}

\footnotetext{We note that for the analysis it suffices to restrict reactions only when $M_2$ is activated.}

We begin by analyzing for which agents truthtelling is a dominant strategy:

\begin{claim}\label{clm:truthtelling} If $(v_{max}+4)\epsilon \leq \delta\frac{g}{|S|}$ then truthtelling is dominant for all agents with $v_i \leq v_{max}$.
\end{claim}
\begin{proof}
We first analyze the effect of misreporting in $M_1$ and $M_2$:

\paragraph{Misreporting in $M_1$}

As $u_i^\out(t_i,s,r) \in [-1,1]$ we can use the simple corollary following Lemma~\ref{lem:indifference} and get that for all possible declarations of the other agents $t'_{-i}$ and all $t'_i$:
\begin{eqnarray*}
\expectation_{s\sim M_1(t'_{-i}, t'_i)}[u_i^\out(t_i,s,r_i(t'_i,s))] - 
   \expectation_{s\sim M_1(t'_{-i}, t_i)}[u_i^\out(t_i,s,r_i(t_i,s))] & \leq & \\
\expectation_{s\sim M_1(t'_{-i}, t'_i)}[u_i^\out(t_i,s,r_i(t'_i,s))] - 
   \expectation_{s\sim M_1(t'_{-i}, t_i)}[u_i^\out(t_i,s,r_i(t'_i,s))] & < & 4\epsilon,
\end{eqnarray*}
where the first inequality follows from $u_i^\out(t_i,s,r_i(t'_i,s)) \leq u_i^\out(t_i,s,r_i(t_i,s))$.
In words, misreporting can gain at most $4\epsilon$ in the expected $u_i^\out$.\footnote{Similarly, even if reactions are not restricted when $M_1$ is activated we get that: $\expectation_{s\sim M_1(t'_{-i}, t'_i)}[u_i^\out(t_i,s,r_i(t_i,s))] - \expectation_{s\sim M_1(t'_{-i}, t_i)}[u_i^\out(t_i,s,r_i(t_i,s))] < 4\epsilon$. We only need reactions to be restricted when $M_2$ is activated.}
Noting that misreporting can gain agent $i$ at most $v_i\cdot\epsilon$ in $u_i^\inf$, we get that the total gain in utility due to misreporting by agents with $v_i \leq v_{max}$ is $(v_{max}+4)\epsilon$.

\paragraph{Misreporting in $M_2$}

If $t'_i \neq t_i$ then with probability at least $\frac{1}{|S|}$ we get that $r_i(t'_i,s)\not= r_i(t_i,s)$. Since the mechanism restricts agent $i$'s reaction to $r_i(t'_i,s)$ we get that 
$$\expectation_{s\sim M_2(t'_{-i}, t_i)}[u_i^\out(t_i,s,r_i(t_i,s))] - 
   \expectation_{s\sim M_2(t'_{-i}, t'_i)}[u_i^\out(t_i,s,r_i(t'_i,s))] \geq \frac{g}{|S|},$$
where $g$ is the minimal utility gap due to not acting according to the optimal reaction. Note that, as $M_2$ ignores its input, misreporting does not yield a change in $u_i^\inf$. We get that the total loss in utility in $M_2$ due to misreporting is at least $g/|S|$.

We get that if $(v_{max}+4)\epsilon \leq \delta \frac{g}{|S|}$ then overall gain in utility due to misreporting is negative for all agents with $v_i \leq v_{max}$, hence truthtelling is dominant for these agents.
\end{proof}

Let $\opt(t) = \max_{s\in S} f(t,s)$ be the optimal value for $f$. We next show that our mechanism approximately recovers $\opt(t)$. 

\begin{claim}\label{clm:accuracy} If $(v_{max}+4)\epsilon \leq \delta \frac{g}{|S|}$ then $$\expectation_{s\sim M(t')]}[ f(t,s) ]\geq \opt(t)-\Delta f\cdot\left(\delta n + 2n^{\alpha}+  2\ln(n|S|)/\epsilon\right) \ . $$
\end{claim}
\begin{proof}
Define $\opt' = \max_s f(t',s)$. Denote by $\bar t_k$ the vector constructed from the $k$ first entries of $t$ and the $n-k$ last entries of $t'$. For all $s$ we have that 
$$
f(t',s) = f(\bar t_0,s) = \sum_{k=0}^{n-1} \left(f(\bar t_k, s) - f(\bar t_{k+1},s)\right) + f(t,s).
$$ 
Note that by Claim~\ref{clm:truthtelling} $t'_i\not= t_i$ for at most $n^\alpha$ entries, and hence $f(\bar t_k, s) - f(\bar t_{k+1},s)\not=0$ for at most $n^\alpha$ values of $k$, in which case it is upper bounded by $\Delta f$. We get hence that $\opt' \geq \opt(t) - n^\alpha \Delta f$.

We get that $M_1(t')$ outputs $s'$ such that $f(t',s') < \opt' - 2\Delta f \ln(n|S|)/\epsilon$ with probability 
$$\frac{\exp(\epsilon f(t',s')/2 \Delta f)}{\sum_{s\in S} \exp(\epsilon f(t',s)/2 \Delta f)} \leq 
\frac{\exp(\epsilon (\opt'-2\Delta f\ln(n|S|)/\epsilon)/2 \Delta f)}{\exp(\epsilon \opt'/2\Delta f)} = \frac{1}{n|S|}.$$
Using the union bound (over elements of $S$), and the fact that $\opt'\leq n \Delta f$, we get a lower bound on the expected revenue of $M_1$ as follows:
\begin{eqnarray*}
\expectation_{s\sim M_1(t')} [f(t,s)] & \geq & (\opt'-2\Delta f \ln(n|S|)/\epsilon)\left(1-|S|\frac{1}{n|S|}\right) \\
& \geq & \opt'-2\Delta f \ln(n|S|)/\epsilon - \Delta f \\
& \geq &\opt(t)-2n^{\alpha}\Delta f -  2\Delta f \ln(n|S|)/\epsilon.
\end{eqnarray*}

We conclude that 
\begin{eqnarray*}
\expectation_{s\sim M(t')} [f(t,s)] & \geq & (1-\delta) \expectation_{s\sim M_1(t')} [f(t,s)] \\
& \geq & (1-\delta)\left(\opt(t)-2n^{\alpha}\Delta f -  2\Delta f \ln(n|S|)/\epsilon \right) \\
& \geq & \opt(t)-\delta n \Delta f - 2n^{\alpha}\Delta f -  2\Delta f\ln(n|S|)/\epsilon.
\end{eqnarray*}
\end{proof}

Setting $\epsilon = n^{-(1+\alpha)/2}\sqrt{g\ln( n|S|)/|S|}$ and $\delta = 2 n^{(\alpha-1)/2}\sqrt{|S|\ln(n|S|)/g}$ we get 
\begin{theorem}\label{thm:main}
Let $n$ be the number of agents, $T$ be a finite type set and $S$ a finite set of alternatives. Let $f: T^n \times S \rightarrow \re$ be an objective function with sensitivity $\Delta f$ and $M$ be the mechanism described in~Algorithm~\ref{alg:one}.

If $\alpha$ is such that the agent population can be assumed to be $(\alpha,1-\alpha)$-admissible, then $M$ recovers $\opt(t)$ to within additive difference of $O\left(\Delta f n^{(1+\alpha)/2} \sqrt{|S| \ln (n|S|)/g}\right)$.
\end{theorem}

The relative accuracy of our mechanisms, in the sense of the difference between the optimal value when agents are privacy-aware or not, increases with larger populations. As described before, natural distributions of the privacy valuations will be $(\alpha,1-\alpha)$-admissible even for very small values of $\alpha$, and therefore the dominating term in the expression in Theorem~\ref{thm:main} can be made arbitrarily close to $\tilde{O}(\sqrt{n})$.


\subsection{Example: Privacy-Aware Selling of Digital Goods}\label{sec:mechanism digital goods}

We describe now an example of a natural privacy-aware mechanism that naturally falls within our framework.

\begin{example}[Pricing a Digital Good]\label{example:pricing digital good}
An auctioneer selling a digital good wishes to design a single price mechanism that would (approximately) optimize her revenue. Every party has a valuation $t_i \in Q=\{0,\frac{1}{q},\frac{2}{q},\ldots,1\}$ for the good (for some constant $q$), and a privacy preference $v_i$. Agents are asked to declare $\tau_i=(t_i,v_i)$ to the mechanism, which chooses a price $p$ for the good. Denote by $\tau'_i=(t'_i,v'_i)$ the actual declaration of agent $i$. If $t'_i\geq p$ then agent $i$ receives the good and pays $p$, otherwise, agent $i$ learns $p$ but does not pay nor receive the good. Agents prefer receiving the good to not receiving it.
\begin{itemize}
\item The utility $u^\out$ is the `traditional' utility, i.e., zero if agent $i$ does not receive the good, and $t_i-p+\frac{1}{2q}$ otherwise, where the additive $\frac{1}{2q}$ is used for modeling preference to receive the good. 
\item For $u^\inf$, we assume that whether agent $i$ received the good and paid for it can be kept completely hidden from all other parties (this can be implemented using cryptographic techniques). Hence, only leakage due to making $p$ public affects $u^\inf$.
\end{itemize}
\end{example}

Consider now the auctioneer from Example~\ref{example:pricing digital good}, and assume that the valuations $t_i$ are taken from $T=\{0,\frac{1}{q},\frac{2}{q},\ldots,1\}$ and similarly, the price $p\in S=\{\frac{1}{q},\frac{2}{q},\ldots,1\}$ for some integer constant $q>1$. Let $\alpha$ be the smallest value such that the agent population can be assumed to be $(\alpha,1-\alpha)$-admissible.

Defining the reactions to be $\{\mbox{\rm \bf buy}, \mbox{\rm \bf not buy}\}$ and optimal reactions $r_i(t,p)= \mbox{\rm \bf buy}$ if $t\geq p$ and $\mbox{\rm \bf not buy}$ otherwise we get that the gap $g$
is $1/2q$.

Suppose the designer goal is to recover the optimal revenue, i.e., $\max_{p\in S} f(t,p)$ where $f(t,p) = p\cdot|\{i: t_i \geq p\}|$ and note that $\Delta f=1$. 

Using Theorem~\ref{thm:main} we get a privacy-aware mechanism that recovers the optimal revenue to within additive difference of 
$O\left(\Delta f n^{(1+\alpha)/2}\sqrt{|S|\ln(n|S|)/g}\right) = O\left(n^{(1+\alpha)/2} q \sqrt{\ln(nq)}\right)$.

Note that the accuracy of this privacy aware mechanism is only (essentially) a factor $\tilde{O}(n^{\frac{\alpha}{2}})$ away from the similar (non-privacy aware) mechanism from~\cite{NST12}. 





\bibliographystyle{alpha}
\newcommand{\etalchar}[1]{$^{#1}$}

\appendix

\section{Omitted Proofs}\label{Missing Proofs} 

\subsection{Proof of Lemma~\ref{lem:indifference}}

\begin{proof} 
Let $t,t',g$ be as in the lemma.
$$
\expectation_{s\sim M(t)} [g(s)] = \sum_{s\in S} M(t)(s) g(s) \leq \sum_{s\in S} e^{\epsilon} M(t')(s) g(s) = e^{\epsilon} \cdot \expectation_{s\sim M(t')} [g(s)],
$$
where the inequality follows since $M$ provides $\epsilon$-differential privacy, and $g$ is non-negative. For $\epsilon \leq 1$ and $g: S\rightarrow [0,1]$ we get 
$$
\expectation_{s\sim M(t)} [g(s)] - \expectation_{s\sim M(t')} [g(s)] \leq (e^{\epsilon}-1) \cdot \expectation_{s\sim M(t')} [g(s)]  \leq e^{\epsilon}-1,
$$
where the last inequality holds because $g$ returns a values in $[0,1]$. Similarly, we get
$\expectation_{s\sim M(t')} [g(s)] - \expectation_{s\sim M(t)} [g(s)] \leq e^{\epsilon}-1$, hence 
$$\left|\expectation_{s\sim M(t)} [g(s)] - \expectation_{s\sim M(t')} [g(s)]\right| \leq e^{\epsilon}-1 < 2\epsilon,$$ where the last inequality follows noting that $(e^{\epsilon}-1) \le 2\epsilon$ for $\ 0 \leq \epsilon \leq 1$.
\end{proof}

\subsection{Proof of Theorem~\ref{thm:expdp}}

\begin{proof}
Let $t,t'$ be neighboring, and $S' \subseteq S$. 
\begin{eqnarray*}
M_f^{\frac{\epsilon}{2\Delta f}}(t)(S') & = & \sum_{s\in S'} \frac{\exp(\frac{\epsilon}{2\Delta f} f(s,t))}{\sum_{s'\in S} \exp(\frac{\epsilon}{2\Delta f} f(s',t))} \\
& = & \sum_{s\in S'} \frac{\exp(\frac{\epsilon}{2\Delta f}(f(s,t) - f(s,t'))) \exp(\frac{\epsilon}{2\Delta f} f(s,t'))}{\sum_{s'\in S} \exp(\frac{\epsilon}{2\Delta f}((f(s',t) - f(s',t')))\exp(\frac{\epsilon}{2\Delta f} f(s',t')))} \\
& \leq & \sum_{s\in S'} \frac{\exp(\frac{\epsilon}{2}) \exp(\frac{\epsilon}{2\Delta f} f(s,t'))}{\sum_{s'\in S} \exp(-\frac{\epsilon}{2})\exp(\frac{\epsilon}{2\Delta f} f(s',t')))} \\
& = & \exp(\epsilon) M_f^\epsilon(t)(S'),
\end{eqnarray*}
where the inequality follows by recalling that $\Delta f \geq |f(s,t)-f(s,t')|$ for all $s,t,t'$.
\end{proof}

\end{document}